\documentclass[11pt,onecolumn]{article}
\usepackage{times,amsmath,epsfig}
\usepackage{xspace,latexsym,syntonly}
\usepackage{amssymb}
\usepackage{amsfonts}
\usepackage{textcomp}
\usepackage{graphicx}
\usepackage{epsfig}
\newcounter{theorem}

\newtheorem{lemma}[theorem]{Lemma}
\newtheorem{claim}[theorem]{Claim}
\newcounter{definition}
\newtheorem{definition}{Definition}

\newcommand{\beq}{\begin{equation}}
\newcommand{\eeq}{\end{equation}}
\newcommand{\bea}{\begin{array}}
\newcommand{\ena}{\end{array}}
\newcommand{\bdf}{\begin{definition}}
\newcommand{\blm}{\begin{lemma}}
\newcommand{\edf}{\end{definition}}
\newcommand{\elm}{\end{lemma}}
\newcommand{\bthm}{\begin{theorem}}
\newcommand{\ethm}{\end{theorem}}
\newcommand{\bprp}{\begin{prop}}
\newcommand{\eprp}{\end{prop}}
\newcommand{\bcl}{\begin{claim}}
\newcommand{\ecl}{\end{claim}}
\newcommand{\bcr}{\begin{coro}}
\newcommand{\ecr}{\end{coro}}
\newcommand{\bquest}{\begin{question}}
\newcommand{\equest}{\end{question}}

\newcommand{\larrow}{{\larrow}}

\newcommand{\nin}{{\not \in}}

\newenvironment{proof}{\textit{Proof:}}{\hfill$\square$\\}
\newtheorem{observation}{Observation}
\newtheorem{theorem}{Theorem}
\begin{document}
\title{ Expected time complexity of the auction algorithm and the push relabel algorithm for maximal bipartite matching on random graphs }
\author{Oshri Naparstek and Amir Leshem}
\date{\today}
\maketitle
\begin{abstract}
In this paper we analyze the expected time complexity of the auction algorithm for the matching problem on random bipartite graphs. We prove that the expected time complexity of the auction algorithm for bipartite matching is $O\left(\frac{N\log^2(N)}{\log\left(Np\right)}\right)$ on sequential machines. This is equivalent to other augmenting path algorithms such as the HK algorithm. Furthermore, we show that the algorithm can be implemented on parallel machines with $O(\log(N))$ processors and shared memory with an expected time complexity of $O(N\log(N))$.

\end{abstract}

\section {Introduction}
 One of the most extensively studied problems in combinatorial optimization in the last 50 years known as the Bipartite Maximum Cardinality Matching (BMCM). The main goal of the BMCM problem is to find an assignment where the maximal number of vertices are matched on bipartite graphs, thus making it a special case of the max-flow problem as well as the max-sum assignment problem. Hence, any algorithm that solves one of these problems also solves the BMCM problem. Most of the algorithms that solve the BMCM problem are augmenting paths algorithms.The Hopcroft-Karp (HK) algorithm \cite{hopcroft1973n} is one of the most well-studied of these. The HK algorithm has a worst case time complexity of $O\left(\sqrt{N}m\right)$ where $N$ is the number of vertices in the bipartie graph and $m$ is the number of edges in the graph. The HK algorithm was revisited more recently by Feder and Motwani \cite{feder1991clique} and was proven to have a worst case time complexity of $O\left(\frac{\sqrt{N}m}{\log(N)}\right)$. Another well known augmenting path algorithm that solves the BMCM problem is Dinic's algorithm \cite{dinic1970algorithm}. Motwani \cite{motwani1994average} proved that Dinic's algorithm achieves perfect matching on random bipartite graphs with an expected running time of $O\left(\frac{m\log(N)}{\log\Delta}\right)$ where $\Delta$ is the expected number of neighbours per vertex. It was later shown by Bast et al. \cite{bast2006matching} that algorithms that use the shortest augmenting paths, have an expected time complexity of $O\left(\frac{\sqrt{N}m}{\log(N)}\right)$. Recently it was shown by Frize et al. \cite{chebolu2010finding} that maximum cardinality matching can be found on sparse random graphs with an expected time complexity of $O\left(m\right)$ using a combination of the Karp-Sipser heuristic and an augmenting path. Goel et al. \cite{goel2010perfect} presented an algorithm with an expected time of $N\log(N)$ iterations on regular bipartite graphs.

 Other solutions to the BMCM problem that do not use the augmenting path method include the auction algorithm \cite{bertsekas1979distributed} and the push relabel (PR) algorithm \cite{goldberg1988new}. The main difference between the augmenting path algorithms described above and these algorithms is that in augmenting path algorithms, an augmenting path is first found and then augmented. By contrast, in the auction and PR algorithms, only two edges are augmented on each iteration according to some update rule. It was shown by both Goldberg \cite{goldberg1995efficient} and Bertsekas \cite{bertsekas1992forward} that the push relabel algorithm and the auction algorithm are equivalent. Since these algorithms are equivalent we will refer only to the auction algorithm from now on.
 The average time complexity of the auction algorithm is not known and the problem of calculating it has remained unresolved  more than $30$ years. However, in many cases it has been shown that the auction algorithm converges faster than other methods in solving the assignment problem \cite{bertsekas1990auction}. Furthermore, several experimental studies \cite{setubal1993new,setubal1996sequential,cherkassky1998augment,kaya2012push} showed that in practice, auction algorithms outperforms augmenting paths based algorithms on many real life scenarios for the solution of the BMCM problem. Bertsekas \cite{bertsekas1991reverse} conjectured that the average running time of the auction algorithm would be $O\left(m\log(N)\right)$ on bipartite graphs with uniformly distributed weights where m is the number of edges in the bipartite graph and $N$ is the number of vertices on each side of the graph.

In this paper we analyze the average time complexity of the auction algorithm for the BMCM problem.
We prove that the expected time complexity of the auction algorithm for random bipartite graphs where each edge is independently selected with probability $p\geq \frac{2\log(N)}{N}$ is $O\left(\frac{N\log^2(N)}{\log\left(Np\right)}\right)$  on sequential machines. We show that by reducing the density of the graph such that $p=O\left(\frac{c\log{N}}{N}\right)$ the complexity reduces to $O\left(\frac{N\log^2(N)}{\log\left(\log(N)\right)}\right)$. We then present a parallel implementation of the algorithm for parallel machines with $O\left(\log(N)\right)$ processors and a shared memory. We prove that the expected time complexity of the parallel implementation is $O(N\log(N))$.

\section {Notation and problem formulation}

\begin{definition}
Let $G=(V,E)$ be a graph with a vertex set $V$ and an edge set $E$. The neighbor set of vertex $v\in V$ is given by
\beq
n_v=\left\{u \in V: \{u,v\}\in E\right\}
\eeq
\end{definition}
\begin{definition}
Let $G=(U,V,E)$ be bipartite graph with vertex sets $|U|=|V|=N$ and an edge set $E$. Let $M\subseteq E$ and let $\tilde{G}(\tilde{U},\tilde{V},M)$ be a bipartite subgraph of $G$ with vertex sets $|\tilde{U}|,|\tilde{V}|=N$ and an edge set $M$. We say that $M$ is a matching on $G$ if
\beq
\max_{v\in U\cup V}|n_v|=1
\eeq
\end{definition}
\begin{definition} Let $G=(U,V,E)$ be bipartite graph with vertex sets $|U|=|V|=N$ and an edge set $E$. If $M$ is a matching and $|M|=N$ then we say $M$ is a perfect matching.
\end{definition}
\begin{definition}
Let $G=(U,V,E)$ be bipartite graph with vertex sets $|U|=|V|=N$ and an edge set $E$. Let $M\subseteq E$ and let $\tilde{G}(U,V,M)$ be a bipartite subgraph of $G$ with vertex sets $|U|,|V|=N$ and an edge set $M$.
We say that a vertex $v\in U \cup V $ is free if $|n_v|=0$ otherwise we say it is not free.
\end{definition}
\begin{definition}
 Let $G=(V,E)$ be a graph with vertex set $V$ and an edge set $E$. A \emph{path} $P$ of length $l$ is an ordered set with $l+1$ vertices $\{v_1,v_2,..,v_{l+1}\}$ such that $v_i\in V,\forall i=1..l+1$ and $(v_i,v_{i+1})\in E ,\forall i=1..l$.
\end{definition}
\begin{definition}
Let $G=(U,V,E)$ be bipartite graph with vertex sets $|U|=|V|=N$ and an edge set $E$. Let $M$ be a non-maximal matching on $G$.  An \emph{alternating path} $P$ on $G$ given partial matching $M$ is a path where $(v_i,v_{i+1})\in M$ if $i$ is even and $(v_i,v_{i+1})\notin M$ if $i$ is odd. An illustration of an alternating path is shown in Fig \ref{figure_alternating_path}.
\end{definition}
\begin{figure}[htbp]
\centering \includegraphics[width=0.3\textwidth]{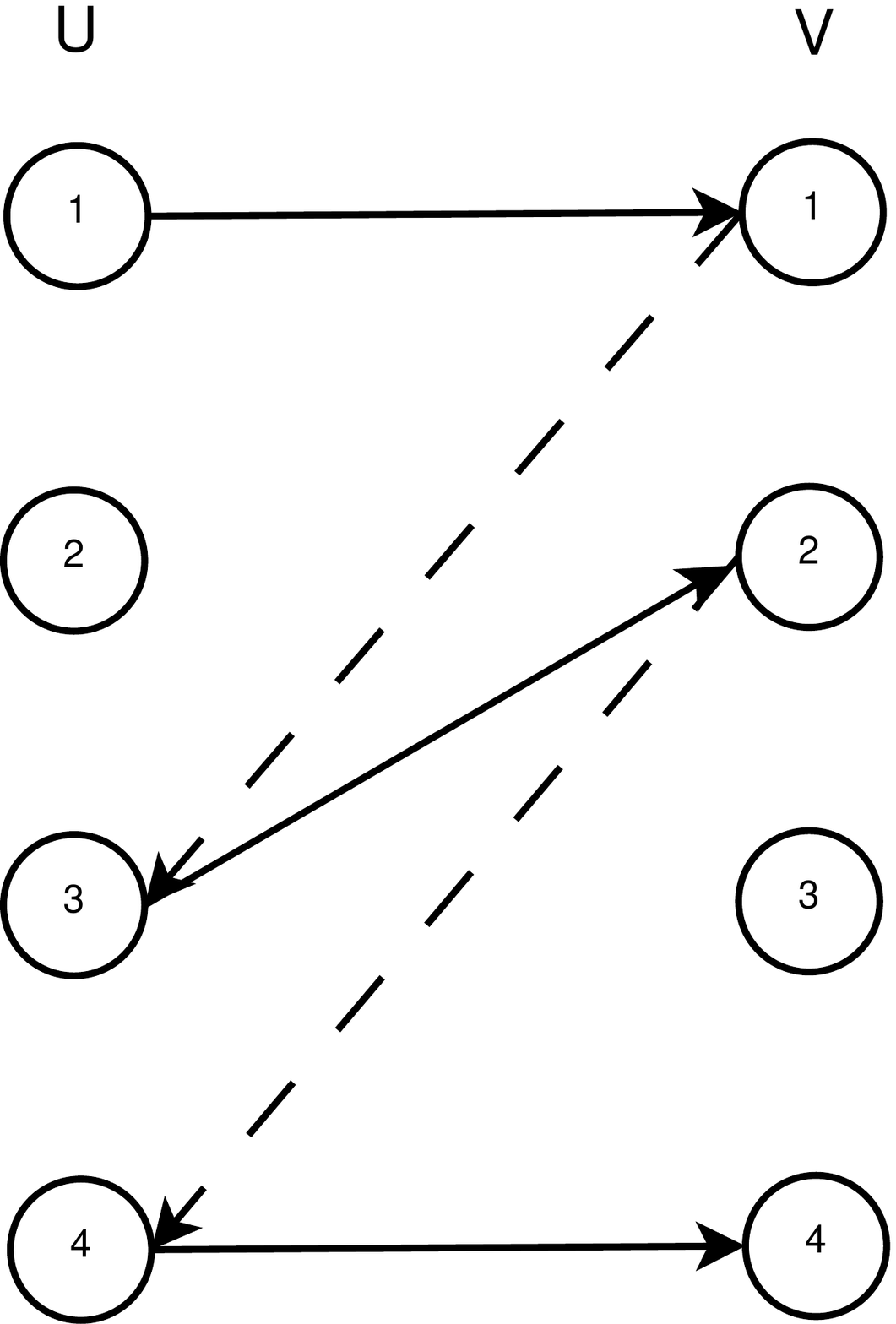}
\caption{Alternating path \ref{lemma_neighbor} }
\label{figure_alternating_path}
\end{figure}
\begin{definition}
Let $G=(U,V,E)$ be a bipartite graph with vertex sets $|U|=|V|=N$ and an edge set $E$. Let $M$ be a non-maximal matching. An \emph{augmenting path} of $M$ is an alternating path on $G$ that starts at a free vertex in $U$ and ends at a free vertex in $V$.
\end{definition}
The BMCM problem is defined as follows:
 Let $G=(U,V,E)$ be a bipartite graph with vertex set $U\cup V$ and an edge set $E$. Find an assignment $M$ with maximal number of edges.

\section{The auction algorithm}
\label{section_auction_alg}
The auction algorithm \cite{bertsekas1979distributed} is an intuitive method for solving the assignment problem. In many cases it has been shown to converge faster than other methods for this case \cite{bertsekas1990auction}. Auctions in which unassigned people raise their prices and bid for objects simultaneously was the original inspiration for the auction algorithm. Similarly, the auction algorithm has two stages, the bidding stage and the assignment stage. In the bidding stage each unassigned individual raises the price of the object he wishes to acquire by the difference between the most profitable object and the second most profitable object plus some constant $\epsilon$. In the assignment stage every object is assigned to the highest bidder. The two stages are repeated until all bidders are assigned an object. More specifically, let $\textbf{R}$ be the matrix of the initial rewards. Let $\textbf{B}$ be a matrix of the bids. $\boldsymbol\rho$ is the price vector where $\rho_k$ is given by:
\beq\label{eq_rho}
\rho_k=\max_{n}\textbf{B}(n,k)
\eeq
Let $\boldsymbol\eta=\left[\eta_1,\eta_2,...,\eta_N\right]$ be an assignment (permutation) vector  where $\eta_n$ is the object that is assigned to the $n$'th person; i.e., the matching $\{(U_n, V_{\eta_n}):n=1..N\}$ is a perfect matching.
\beq
(U_n,V_{\eta_n})\in M
\eeq
\begin{definition} An object $k$ is said to be \emph{assigned} to person $n$ by $\boldsymbol\eta$ if $\boldsymbol\eta_n=k$.
\end{definition}
\begin{table}
\caption{The Auction algorithm (Bertsekas 79)}
\label{table_auction_alg}
\begin{tabular}{l}
\hline
Select $\epsilon>0$, set all the people as unassigned and set $\rho_k=0,n=1..N$\\
\textbf{Repeat until all the people are assigned}\\
\ \ 1. Choose an unassigned person $n$\\
\ \ 2. Calculate his maximum profit $\gamma_{n}=\max_k(\textbf{R}(n,k)-\rho_k)$ \\
\ \ 3. Calculate the second maximum profit\\
\ \ \ \ \ $\tilde{k}=\arg\max_k(\textbf{R}(n,k)-\rho_k)$\\
\ \ \ \ \ $\omega_{n}=\max_{k\neq \tilde{k}}(\textbf{R}(n,k)-\rho_k)$\\
\ \ 4. Assign object $\tilde{k}$ to person $n$. If this object has been\\
\ \ \ \ \ assigned to another person, make this \\
\ \ \ \ \ person unassigned (and as a result unassigned).\\
\ \ 5. Set person $n$ as assigned\\
\ \ 6. Update the price of object $\tilde{k}$ to be \\
\ \ \ \ \ $\rho_{\tilde{k}}=\rho_{\tilde{k}}+\gamma_{n}-\omega_{n}+\epsilon$\\
\textbf{end} \\
\hline
\end{tabular}
\end{table}
The reward of the $n$'th person on assignment $\boldsymbol\eta$ is denoted by $\textbf{R}(n,\eta_n)$
and the price that the $n$'th person pays on assignment $\boldsymbol\eta$ is denoted by
$\rho_{\boldsymbol\eta_{n}}$.
Given a positive scalar $\epsilon$, an assignment $\eta$ and a price $\rho_{\eta_n}$, a person $n$ is termed \emph{happy} with assignment $\boldsymbol\eta$ if the profit (i.e., reward minus price) is within $\epsilon$ of the maximal profit achievable by person $n$. This condition is called $\epsilon$-Complementary slackness ($\epsilon$-CS).
\beq\label{eq_ecs}
\textbf{R}(n,\eta_n)-\rho_{\eta_n}\geq\max_k(\textbf{R}(n,k)-\rho_k)-\epsilon
\eeq
When all the people are assigned and happy, the algorithm stops. It was shown in \cite{bertsekas1979distributed} that the algorithm terminates in finite time. The algorithm is within $N\epsilon$ of being optimal at termination \cite{bertsekas1979distributed}. Also, if the initial prices are all zeros and $N\leq K$ the algorithm still converges in finite time and with the same bounds on optimality. The original Bertsekas auction algorithm is depicted in Table \ref{table_auction_alg}.
\section{A simplified auction algorithm for maximal bipartite matching}
The BMCM problem is a special case of the assignment problem where the values in the reward matrix are either $0$ or $1$. Furthermore, since the reward matrix is integer valued, the choice of $\epsilon= \frac{1}{N}$ is sufficient for obtaining an optimal solution \cite{bertsekas1979distributed}. The cardinality of a matching is the number of vertices that were assigned to edges with weight $1$.
Here we simplify the auction algorithm and assume that on each iteration an unassigned person is picked and raises the price of his most desirable object by exactly $\epsilon$.  Under this assumption, the price
\beq
\rho_k=\frac{h_k}{N}
\eeq
where $h_k$ is the number of times that the price of the $k$'th object was raised.
We further simplify the algorithm and assume that a person would only bid on objects with positive rewards.
Under this assumption the $\epsilon$-CS condition \ref{eq_ecs} for an object with a positive reward becomes:
\beq
1-\frac{h_{\eta_n}}{N}\geq\max_k(1-\frac{h_k}{N})-\frac{1}{N}.
\eeq
Multiplying both sides by $N$ and rearranging yields a simplified $\epsilon$-CS condition for the BMCM problem
\beq\label{eq_ecs_bipart}
h_{\eta_n}-1\leq\min_k(h_k).
\eeq
Using the simplified $\epsilon$-CS condition we can derive a simplified version of the auction algorithm for the BMCM problem.
On each iteration an unassigned person is chosen, the person is assigned to the $k$'th channel with minimal $h_k$ and raises it by $1$.
The description of the algorithm using graph theory terminology is as follows:
Let $G=(U,V,E)$ be a bipartite graph with vertex sets $U,V$ where $|U|=|V|=N$ and an edge set $|E|$. Let $\textbf{h}=[h_1,h_2,...,h_N]$ be the value assigned to the vertices in $V$. On each iteration a free vertex $u\in U$ is  assigned to a vertex $v \in V$ with minimal value $h_k$. If another vertex was previously assigned to that vertex it becomes free. The algorithm stops after a perfect matching is found or the number of iterations is larger than $N^2+N$.
The algorithm is depicted in Table \ref{table_alg}. Note that the simplified auction algorithm is equivalent to the push relabel algorithm with double push \cite{goldberg1995efficient}.
\begin{table}
\label{table_alg}
\caption{Simplified auction algorithm for maximal matching in bipartite graphs}
\begin{enumerate}
\item Initialize $h_v=0,\forall v\in V$ and set $M=\emptyset$
\item While $|M|<N$ and $\sum_{k=1}^Nh_k<N(N-1)$ do
\begin{enumerate}
\item Choose a free vertex $u\in U$
\item $j=\arg\min_{v\in n_{u}}h_v$
\item $M= M\cup (u,j)$
\item $u_{old}=\left\{u \in U:(u,j)\in M\right\}$
\item $M=M\setminus(u_{old},j)$
\item $h_j=h_j+1$
\end{enumerate}
\item Return
\end{enumerate}

\end{table}

\section{The average number of iterations of the simplified auction algorithm for the BMCM problem}
\label{section_alg}
In this section we analyze the expected time complexity of the simplified auction algorithm for the BMCM problem.
\begin{observation} \label{lamma_T_sum_H}Let $T$ be the number of iterations until the algorithm terminates and let $h_v$ be the value of vertex $v$ at termination, then
\beq
T=\sum_{v=1}^N h_v
\eeq
\end{observation}

\begin{proof}
The proof is trivial since on every iteration the value of exactly one vertex is increased by $1$.
\end{proof}
\begin{lemma}
\label{lemma_neighbor}Let $G=(U,V,E)$ be a bipartite graph with vertex sets $|U|=|V|=N$ and an edge set $E$. Let $M(i)\subseteq E$ be a non maximal matching obtained by the algorithm in the $i$'th iteration. Let $h_v(i)$ be the value of vertex $v$ in the $i$'th iteration of the algorithm.
Let $D_l(i)$ be a subset of $V$ defined by:
\beq
D_l(i)=\left\{v\in V:h_v(i)\geq l\right\}.
\eeq
If $v_0\in D_l(i)$ and $(u,v_0)\in M(i)$ then
\beq
n_u\subseteq D_{l-1}(i)
\eeq
\end{lemma}

\begin{proof}
If $v_0\in D_l(i)$ and $(u,v_0)\in M(i)$ it implies that $\displaystyle v_0=\arg\min_{v\in n_u} h_v(k)$ for some $k<i$ and
\beq
h_{v_0}(k+1)=h_{v_0}(k)+1=l.
\eeq
This implies that for every vertex $v\in n_u$ in the $(k+1)$'th iteration
\beq
h_{v}(k+1)\geq h_v(k+1)-1=h_v(k)=l-1.
\eeq
This implies that in the $k$'th iteration
\beq
n_u\subseteq D_{l-1}(k).
\eeq
Since $h_v(n)$ is a non-decreasing function of $n$ (prices never go down)
\beq
n_u\subseteq D_{l-1}(k)\subseteq D_{l-1}(i).
\eeq
\end{proof}
An illustration of lemma \ref{lemma_neighbor} is given in Figure \ref{figure_neghbor}.
\begin{figure}[htbp]
\centering \includegraphics[width=0.45\textwidth]{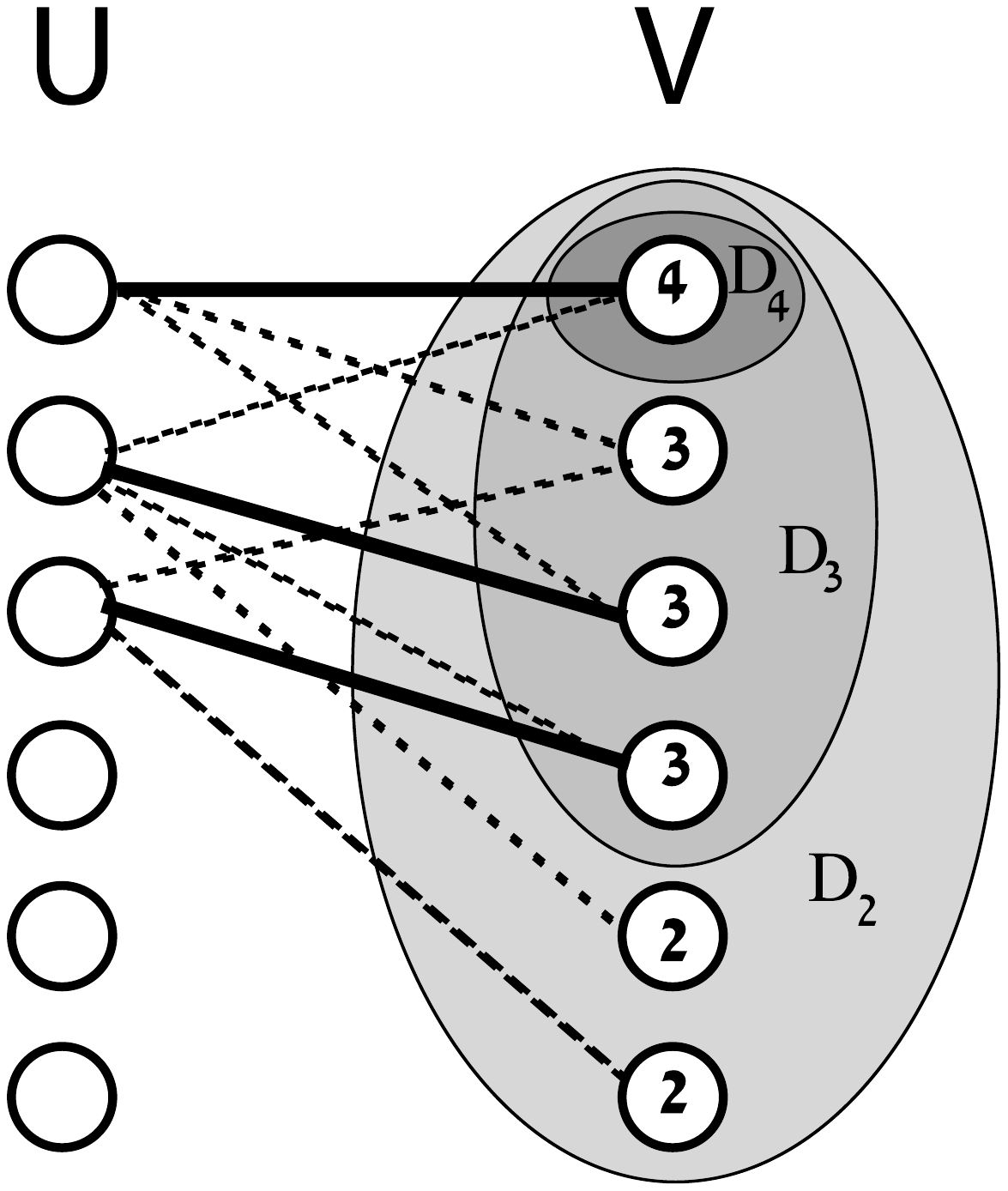}
\caption{An illustration of lemma \ref{lemma_neighbor} }
\label{figure_neghbor}
\end{figure}
%
%
\begin{claim}
\label{lemma_dl_1} Let $G=(U,V,E)$ be bipartite graph with vertex sets $|U|=|V|=N$ and an edge set $E$. Let $M(i)$ be a non maximal matching obtained by the algorithm in the $i$'th iteration. Let $u_0\in U$ be a free vertex such that in the $i$'th iteration of the algorithm $n_{u_0}\subseteq D_l(i)$. Let $u_1\in U$ be the end point of an alternating path $P$ with $|P|=2$ starting from $u_0$ then
\beq
\bea{ll}
v\not\in D_{l-2}(i)\setminus D_{l-1}(i)& \forall v\in n_{u_{1}}
\ena
\eeq
\end{claim}
\begin{proof}
We first observe that if $n_{u_0}\subseteq D_l(i)$, then by definition for every vertex $v\in n_{u_0}$
\beq
v\not\in D_{l-1}(i)\setminus D_{l}(i).
\eeq
By the definition of an alternating path, if $n_{u_0}\subseteq D_l(i)$ there exists a vertex $j\in n_{u_{1}} \cap n_{u_{0}}$ such that $(u_{1},j)\in M(i)$, $(u_{0},j)\nin M(i)$ and $j \in D_l(i) $.
Hence, by lemma \ref{lemma_neighbor} it implies that
\beq
n_{u_{1}}\subseteq D_{l-1}(i).
\eeq
Reapplying the lemma implies that for every vertex $v\in n_{u_{1}}$
\beq
v\not\in D_{l-2}(i)\setminus D_{l-1}(i).
\eeq
\end{proof}

\begin{lemma}
\label{lemma_max_h} Let $G=(U,V,E)$ be bipartite graph with vertex sets $|U|=|V|=N$ and an edge set $E$. Let $M(i)$ be a non maximal matching obtained by the algorithm in the $i$'th iteration and let $u_0\in U$ be a free vertex such that in the $i$'th iteration of the algorithm $n_{u_0}\subseteq D_l(i)$;  then every augmenting path of $G$ on $M(i)$ starting from $u_0$ is at least of length $2l+1$.
\end{lemma}

\begin{proof}
Let $u_m\in U$ be the end point of an alternating path $P$ with $|P|=2m$ starting from $u_0$.
By recursively applying claim \ref{lemma_dl_1} we get that for every vertex $v\in n_{u_{m}}$
\beq\label{eq_D0D1}
v\not\in D_{l-m-1}(i)\setminus D_{l-m}(i).
\eeq
The last vertex in an augmenting path is a free vertex.
Note that a vertex $v\in V$ is free only if
\beq
v\in D_{0}(i)\setminus D_{1}(i).
\eeq
However, from equation \ref{eq_D0D1} if $n_{u_0}\subseteq D_l(i)$ then for all $v\in n_{u_{m}}, m\leq l$,
\beq
\bea{ll}
v\not\in D_{0}(i)\setminus D_{1}(i).
\ena
\eeq
This means that if $P$ is an augmenting path that starts from $u_0$ there are at least $l$ odd elements on $P$ .
This implies that any augmenting path that starts from $u_0$ has at least $2l+1$ elements.
\end{proof}
A well-known theorem by Berge \cite{berge1957two} states that if a bipartite graph $G$ contains a perfect matching there exists an augmenting path in $G$ with respect to any non maximal matching $M$.
\begin{theorem}\cite{berge1957two}
 \label{theorem_augment}  Let $G=(U,V,E)$ be a bipartite graph with vertex sets $|U|=|V|=N$ and an edge set $E$. If $G$ contains a perfect matching there exists an augmenting path in $G$ for any partial matching $M$.
 \end{theorem}
We will now use the above theorem and lemma \ref{lemma_max_h} to prove the following lemma:
\begin{lemma}
 \label{lemma_neighbour_less_l} Let $G=(U,V,E)$ be a bipartite graph with vertex sets $|U|=|V|=N$ and an edge set $E$. If $G$ contains a perfect matching and for any non-perfect matching $M\subseteq E$  there exists an augmenting path of length at most $2l+1$, then for each free vertex $u\in U$ and a partial matching $M(i)$ obtained by the algorithm in the $i$th iteration there exists at least one neighbor $j\in n_u$ such that $h_j(i)\leq l$ on each iteration of the algorithm.
 \end{lemma}
%
\begin{proof}

Assume towards contradiction that $G$ contains a perfect matching and for any non-perfect matching $M\subseteq E$  there exists an augmenting path of length at most $2l+1$ but on the $i$'th iteration there exists a free vertex $u$ with respect to $M(i)$ such that if $j\in N_u$ then $h_j(i)> l$. This means that $n_u\subseteq D_{l+1}(i)$. By lemma \ref{lemma_max_h} each augmenting path that begins at $u$ is of length of at least $2l+3$. However, our assumption is that for any non maximal matching $M$ of $G$ there exists at least one augmenting path of length $2l+1$ or less. This implies that $u$ is not the starting vertex of any augmenting path of $G$ on $M(i)$. This also implies that $u$ is not the starting vertex of any augmenting path of $G$ on $M(j)$ for any $j\geq i$. This is true since we know from lemma \ref{lemma_neighbor} that
\beq
D_{m}(i)\subseteq D_{m}(j),\forall m,\forall j\geq i
\eeq
 Hence, $u$ is not the starting vertex of an augmenting path of $G$ on $M(j)$ where $|M(j)|=N-1$. This implies that there is no augmenting path of length at most $2l+1$ to $M(j)$. This can occur in two situations:
  \begin {enumerate}
  \item No augmenting path exists to $G$ on $M(j)$. Theorem \ref{theorem_augment} implies that no perfect matching exists for $G$ in contradiction to the assumption that a perfect matching exists.
  \item There exists an augmenting to $G$ on $M(j)$ but with length at least $2l+3$. This is a contradiction to our assumption that every non maximal matching contains an augmenting path of length at most $2l+1$
  \end{enumerate}
\end{proof}
\begin{lemma}
\label{lemma_h_less_2l} Let $G=(U,V,E)$ be a bipartite graph with vertex sets $|U|=|V|=N$ and an edge set $E$. If $G$ contains a perfect matching and for any non-perfect matching $M\subseteq E$ there exists an augmenting path of length at most $2l+1$, then for every $v\in V$ at each iteration of the algorithm until the algorithm terminates
\beq
h_v(i)\leq l+1 ,\forall i,v\in V
\eeq
\end{lemma}
\begin{proof}
On each iteration $i$, the free vertex from $U$ select a vertex $j$ such that $j=\arg\min_{v\in n_u}h_v(i)$ and set $h_j(i+1)=h_j(i)+1$. However, we assume that $G$ contains a perfect matching and for any non-perfect matching $M\subseteq E$ there exists an augmenting path of length at most $2l+1$. Hence, by lemma \ref{lemma_neighbour_less_l} to any free vertex at the $i$'th iteration with respect to a partial matching $M(i)$ there exists a vertex $v_0$ such that $h_{v_0}(i)\leq l$. This means
\beq
\min_{v\in n_u}h_v(i)\leq l,\forall i.
\eeq
As a result, if there exists vertex $v_0\in V$ where $h_{v_0}(i)=l+1$ then
\beq
\arg\min_{v\in n_u}h_v(i)\neq v_0.
\eeq
This implies that if $h_{v_0}(i)=l+1$ its value never rises and as a result
\beq
h_v(i)\leq l+1 ,\forall i,v\in V.
\eeq
\end{proof}
\begin{lemma}
 \label{lemma_worst_case}Let $G=(U,V,E)$ be a bipartite graph with vertex sets $|U|=|V|=N$ and an edge set $E$. If $G$ contains a perfect matching then
\beq
T\leq N(N-1)
\eeq
\end{lemma}
\begin{proof}
$G$ contains $2N$ vertices. Hence, every path is of length at most $2N-1$ and in particular every augmenting path is of length at most $2N-1$ for any non maximal matching.
By lemma \ref{lemma_h_less_2l} we get that
\beq
h_v(i)\leq N-1 ,\forall i,v\in V
\eeq
Using lemma \ref{lamma_T_sum_H} we get:
\beq
T=\sum_{v=1}^N h_v\leq N(N-1)
\eeq
\end{proof}
\begin{definition}   Let $G=(U,V,E)$ be bipartite graph with vertex sets $|U|=|V|=N$ and an edge set $E$. If $G$ is a random graph where each edge occurs with probability $p$ we say that $G\in B(N,p)$.
\end{definition}
\begin{definition}
$G\in\tilde{B}(N,p)$ if $G\in B(N,p)$ and for any non-maximal matching $M$ there exists an augmenting path for $M$ of length at most $2L+1$ where $L=\frac{\tilde {c}\log(N)}{\log(Np)}$ and $\tilde {c}>0$ is some constant.
\end{definition}
\begin{observation} \label{lemma_log_path_alg} Let $G\in \tilde{B}(N,p)$ and let and let $T$ be the number of iterations until the algorithm terminates then
\beq
T\leq N(L+1).
\eeq
where $L=\frac{\tilde {c}\log(N)}{\log(Np)}$.
\end{observation}
\begin{proof}
$G\in \tilde{B}(N,p)$; hence, for any non maximal matching $M$ there exists an augmenting path of length at most $2L+1$.
let $L=\frac{\tilde {c}\log(N)}{\log(Np)}$, then
by lemma \ref{lemma_h_less_2l} we get that
\beq
h_v(i)\leq L+1 ,\forall i , v\in V.
\eeq
Using lemma \ref{lamma_T_sum_H} we get:
\beq
T=\sum_{v=1}^N h_v\leq N(L+1)=\frac{N\tilde {c}\log(N)}{\log(Np)}+N.
\eeq
\end{proof}
The following theorem was proven in \cite{motwani1994average}:
\begin{theorem}\label{lemma_prob_bnp} Let $G\in B(N,p)$ where $p\geq\frac{(1+\epsilon)\log(N)}{N}$ then for every $\gamma>0$ there exists $N_{\gamma}$ such that for every $N\geq N_{\gamma}$
\beq
\Pr(G\in \tilde{B}(N,p))\geq 1-N^{-\gamma}.
\eeq
\end{theorem}
The following theorem was proven in \cite{erdHos1966existence}:
\begin{theorem} \label {theorem_prob_perfect}Let $G=(U,V,E)\in B(N,p)$ and $p=\frac{c\log(N)}{N}, c>2$ then
\beq
\lim_{N\to\infty}\Pr\left(G \textrm{ contains a perfect matching}\right)=e^{-2N^{1-c}}.
\eeq
\end{theorem}
We now prove the main theorem of the paper:
\begin{theorem}\label{theorem_N_iter} Let $G\in B(N,p)$ be a random bipartite graph with $N>N_{\gamma}$ vertices on each side and $p=\frac{c\log(N)}{N},c>2$. Then the expected number of iterations until the algorithm terminates is:
\beq
E(T)\leq O\left(\frac{N\log(N)}{\log(Np)}\right)
\eeq
\end{theorem}
\begin{proof}
Let
\beq
A=\{G\in B(N,p): G \textrm{ does not contains a perfect matching}\}
\eeq
and let
\beq
B=\{G\in B(N,p)\setminus \tilde{B}(N,p)\}.
\eeq
From Theorem \ref{lemma_prob_bnp} we know that there exists $N_{\gamma}$ such that for every $N>N_{\gamma}$
\beq
\Pr\left(B\right)\leq N^{-\gamma}.
\eeq
From theorem \ref{theorem_prob_perfect} we know that if $\left(G\in B(N,p)\right)$, $p>\frac{c\log(N)}{N}$ and $c>2$, there exists $N_1$ such that for every $N>N_1$
\beq
\Pr\left(A\right)\leq \frac{2}{N^{c-1}}\leq \frac{2}{N}.
\eeq
let $\gamma\geq 1$ and let $\tilde{N}$ be
\beq
\tilde{N}=\max(N_{\gamma},N_1)
\eeq
let $\breve{B}(N,p)$  be
\beq
\displaystyle \breve{B}(N,p)=\left\{G:G \textrm{ contains a perfect matching}\right\}\cap\tilde{B}(N,p).
\eeq
Using the union bound,
\beq
\displaystyle \Pr\left(G\in \breve{B}(N,p)\right)=1-\Pr\left(A\cup B\right)\geq 1-\Pr(A)-\Pr(B)
\eeq
If $N\geq\tilde{N}$, $p=\frac{c\log(N)}{n},c>2$ then
\beq
\Pr\left(G\in \breve{B}(N,p)\right)\geq 1-\frac{3}{N}.
\eeq

From lemma \ref{lemma_worst_case} and Lemma \ref{lemma_log_path_alg} we know that:
\beq
T\leq \left\{\begin{matrix} N(N-1),
& G\not\in \breve{B}(N,p)  \\
\frac{N\tilde {c}\log(N)}{\log(Np)}+N,
 & G\in \breve{B}(N,p)
\end{matrix}\right. ,
\eeq

this implies that if $N\geq\tilde{N}$, $p=\frac{c\log(N)}{n},c>2$ then
\beq
\bea{l}
\displaystyle E\left(T\right)\leq \left(1-\frac{3}{N}\right)\left(\frac{N\tilde {c}\log(N)}{\log(Np)}+N\right)+\frac{3}{N}(N^2-N)=\\
\displaystyle = O\left(\frac{N\log(N)}{\log(Np)}\right).
\ena
\eeq
If $p=\frac{c\log(N)}{N}$ the number of iterations is
\beq
O\left(\frac{N\log(N)}{\log(\log(N))}\right).
\eeq
\end{proof}

\section{Expected time complexity}
In this section we use the results from the previous section to analyze the expected time complexity of the algorithm in sequential and parallel implementations. We first prove that the expected running time of the algorithm is $O\left(\frac{N\log^2(N)}{\log(Np)}\right)$ on sequential machines. We point out that if the graph is dense it is advantageous to take a random sparse subgraph of the original graph and perform the algorithm on the sparse subgraph. Then, we introduce a parallel implementation of the algorithm for machines with $Q$ processors and a shared memory. We prove that if $Q=O(\log(N))$ the expected time complexity is $O\left(N\log(N)\right)$.
\subsection{Sequential implementation}\label {subsection_seq}
We now analyze the expected time complexity of the algorithm for random bipartite graphs from $B(N,p)$ on sequential machines. From Theorem \ref{theorem_N_iter} we know that the expected number of iterations is bounded by $O\left(\frac{N\log(N)}{\log(Np)}\right)$. Hence we need to show that the expected numbed of operations per iteration is $O\left(\log(N)\right)$. The next theorem shows that the algorithm can be implemented with  $O\left(\frac{N\log^2(N)}{\log(Np)}\right)$ time complexity on sequential machines.

\begin{theorem}\label{theorem_seq} Let $G=(U,V,E)\in B(N,p)$ be a random bipartite graph with $p\geq\frac{c\log(N)}{N}$,$c>2$ then the algorithm finds a maximal matching on a sequential machine with one processor with the expected time bounded by $O\left(\frac{N\log^2(N)}{\log(Np)}\right)$
\end{theorem}
\begin{proof}

In section \ref{section_alg} we proved that the expected number of iterations until convergence is bounded by
\beq
E\left(T\right)\leq O\left(\frac{N\log(N)}{\log(Np)}\right)
\eeq
where $T$ is the number of iterations until convergence.
Define $T_{inner}$ to be the number of operations needed in each iteration and $T_{total}$ to be the total number of operations preformed by the algorithm.

All of the operations in each iterations are $O(1)$ except for the operation of finding a vertex with minimal value for the chosen vertex $u$ which requires $O(|n_u|)$ operations where $|n_u|$ is the number of neighbours of vertex $u$. Since the number of neighbors of each vertex is an independent random variable and $E(|n_u|)=Np$  then
\beq
E\left(T_{total}\right)=E\left(T\right)E\left(T_{inner}\right)=O\left(\frac{pN^2\log(N)}{\log(Np)}\right)
\eeq
If $E(|n_u|)=O\left(\log(N)\right)$ then
\beq
E\left(T_{total}\right)=O\left(\frac{N\log^2(N)}{\log\log(N)}\right)
\eeq
\end{proof}
When the graph is dense; i.e., $E(|n_u|)=O(N)$ the algorithm converges with a time complexity of $O\left(N^2\log(N)\right)$ which is not particularly good. To improve the expected running time performance, we can obtain a sparse random graph from the dense graph by randomly choosing $d$ edges from the original dense graph where $d$ is binomially distributed as
\beq
d\sim Bin(|E|,\frac{c\log(N)}{Np})
\eeq
Let $\tilde{G}$ be the sparse graph obtained from $G$ by randomly selecting $d$ edges of $G$. If $\tilde{G}$ contains a perfect matching, a solution for the MCM problem for $\tilde{G}$ is also a solution for $G$. If $\tilde{G}$ doest not contain a perfect matching the algorithm is applied on $G$ with $O\left(N^2\right)$ iterations. Hence, the expected time complexity even for dense graphs remains:
\beq
\bea{l}
\displaystyle E\left(T_{total}\right)\leq\left(1-\frac{2}{N}\right)O\left(\frac{N\log^2(N)}{\log\log(N)}\right)+\frac{2}{N}O\left(N^2\log(N)\right)=\\
\displaystyle =O\left(\frac{N\log^2(N)}{\log\log(N)}\right)
\ena
\eeq

\subsection{Parallel Implementation}
We now analyze a parallel implementation of the algorithm on a machine with $Q$ processors and a shared memory. We show that if $Q=O\left(\log(N)\right)$ the expected time complexity of the algorithm is $O\left(N\log(N)\right)$. As in the sequential case, the expected number of outer iterations is $T=O\left(\frac{N\log(N)}{\log(Np)}\right)$. This implies that for $p=\frac{c\log(N)}{N}$ $T=O\left(\frac{N\log(N)}{\log(\log(N))}\right)$. In the parallel implementation we keep a sorted tree for each vertex of $u\in U$. If $G$ is a random graph with $p=\frac{c\log(N)}{N}$ then
\beq
\bea{ll}
E(|n_u|)=c\log(N), & \forall u\in U\\
E(|n_v|)=c\log(N), & \forall v\in V\\
\ena
\eeq
This implies that on each iteration we need to maintain $O\left(\log(N)\right)$ sorted trees with an expected number of $O\left(\log(N)\right)$ elements per tree. The maintenance of each sorted tree can be done in parallel over $Q$ processors with $O\left(\log\left(\log(N)\right)\right)$ operations per processor. Hence, the expected time complexity for the parallel implementation is given by:
\beq
E(T_{total})=E(T)E(T_{inner})=O\left(\frac{N\log^2(N)}{Q}\right)
\eeq
and if $Q=O(\log(N))$ the expected time becomes
\beq
E(T_{total})=O\left(N\log(N)\right)
\eeq

Note that if $G$ is not sparse we can obtain a sparse graph from the dense graph and use the same arguments as in the sequential implementation to get
\beq
E(T_{total})=O\left(N\log(N)\right)
\eeq

\section {Conclusion}
In this paper we analyzed the expected time complexity of the auction algorithm for the matching problem on random bipartite graphs. We proved that the expected time complexity of the auction algorithm for bipartite matching is just as good as other augmenting path algorithms such  as the HK algorithm. Furthermore, we showed that the algorithm can be implemented on parallel machines with $O(\log(N))$ processors and a shared memory with an expected time complexity of $O(N\log(N))$.


\bibliographystyle{amsalpha}
\newcommand{\etalchar}[1]{$^{#1}$}
\providecommand{\bysame}{\leavevmode\hbox to3em{\hrulefill}\thinspace}
\providecommand{\MR}{\relax\ifhmode\unskip\space\fi MR }
\providecommand{\MRhref}[2]{%
  \href{http://www.ams.org/mathscinet-getitem?mr=#1}{#2}
}
\providecommand{\href}[2]{#2}

\end{document}